\documentclass[12pt]{article}
\usepackage{graphicx}
\usepackage{mathptmx,mathrsfs}
\usepackage{amsfonts,amsmath,amssymb,amsthm}
\usepackage{indentfirst}
\usepackage{cite}
\usepackage{hyperref}
\usepackage{enumitem}
\usepackage{bm,graphics,color}
\usepackage[english]{babel}

\numberwithin{equation}{section}
\newtheorem{theorem}{Theorem}[section]
\newtheorem{lemma}[theorem]{Lemma}
\theoremstyle{definition}

\begin{document}

\title{Generalized Stueckelberg-Higgs gauge theory}

\author{C. A. Bonin$^{1}$ \thanks{carlosbonin@gmail.com}, G. B. de Gracia$^{2}$\thanks{gb9950@gmail.com}, A. A. Nogueira$^{3}$\thanks{andsogueira@hotmail.com} and B. M. Pimentel$^{2}$\thanks{bruto.max@unesp.br}
\\
\textit{$^{1}$}\textit{\small{}Departamento de Matem\' atica e Estat\' istica, Universidade Estadual de Ponta Grossa (DEMAT/UEPG),}\\ \textit{\small{} Avenida Carlos Cavalcanti, CEP 4748, 84030-900, Ponta Grossa, PR, Brazil}\\
\textit{$^{2}$}\textit{\small{} Instituto de F\' isica Te\' orica, Universidade Estadual Paulista (IFT/Unesp)}\\
\textit{\small{} Rua Dr. Bento Teobaldo Ferraz 271, Bloco II Barra Funda, CEP
01140-070 S\~ao Paulo, SP, Brazil}\\
\textit{$^{3}$}\textit{\small{} Instituto de F\'isica, Universidade Federal de Goi\' as (IF/UFG),}\\
\textit{\small{}Av. Esperan\c{c}a, CEP 74690-900, Goi\^ania, Goi\' as, Brasil}\\}
\maketitle
\date{}

\begin{abstract}
The aim of this work is to discuss and explore some generalized aspects of  generation of photon mass respecting gauge symmetry. We introduce generalized Stueckelberg and Higgs gauge theories and present the classical and quantum frameworks related to them. We construct the quantum theory by writing the transition amplitude in the Fadeev-Senjanovic formalism and put it in a covariant form by the Fadeev-Popov method. We also analyze the independence of the transition amplitude by gauge parameters via BRS-T symmetry. Even in the generalized context, the Stueckelberg structure influences the quantization process of the Higgs theory in the 't Hooft's  fashion, in which an intimate relationship between the Stueckelberg compensating field and the Goldstone boson arises.

\textbf{Keywords}: Constrained Systems; Gauge Theories
\end{abstract}
\newpage{}

\section{Introductory Aspects}

The current understanding of how to formulate and interpret the laws of nature describing the fundamental interactions between matter and radiation is in terms of gauge theories \cite{Uti, Rai, Leite}. The Quantum Mechanical view of such interactions is through the assumption of the existence of mediating particles. The origin of this idea can be traced back to the (now famous) Yukawa theory \cite{Yukawa}. The core of the Yukawa theory is that the particular interaction conceived in it is short-ranged due to a non-vanishing mass for the corresponding mediating particle. That particle turned out to be the pion, which is associated to a scalar field. However, concerning gauge interactions, extensive studies of a hypothetical mass for vector fields were done by Proca \cite{Proca} and nowadays there are numerous experimental efforts to search for (nonzero) mass for photons \cite{Chinese}.

As the question of mass in the context of gauge symmetry attracted more attention, Stueckelberg explored it with its gauge-symmetry mass generation mechanisms \cite{Stuck, Stueck}. In the present day, Stueckelberg's particle find its application in Cosmology, as an Ultralight Dark Matter Candidate used to explain the anomalous rotation curves of galaxies \cite{Indians}.

It is quite natural to think that these ideas influenced the spontaneous symmetry-breaking mass generation mechanism (of the Higgs type) \cite{Higgs, Coleman, Peskin, Weinb} in which there is an interplay between the gauge fields and the Nambu-Goldstone bosons, whose general result is the gauge field acquiring a non-vanishing mass \cite{Rubakov}. The Higgs model led to the formulation of the Weinberg-Salam electroweak theory, in which massive vector bosons and photons mediate the interaction \cite{Weinberg}.

These concepts paved the way to the understanding of general and important properties related to the mechanisms of mass generation for gauge fields such as unitarity and renormalization \cite{HooftV}. The concept of mass generation also appears in other contexts, like in the Nambu-Jona-Lasinio model \cite{Nambu, Nakanishi} with the dynamical breaking of chiral symmetry \cite{Ebert} and in gauged Thirring model where the gauge symmetry does not prohibit the photon to acquiring mass dynamically, as shown by some of us \cite{Bonin}.

It is
also known that in plasma physics there is a shielding phenomenon where the interaction mediated by the gauge field
becomes short-ranged with the gauge field acquiring a mass - the so-called Debye mass \cite{LeB} and there is also the possibility of exploring the concept of restoration of symmetry (Jackiw-Dolan) above a critical temperature \cite{JD}.

On the other hand Bopp, Podolsky, and Schwed proposed a higher-order derivative Lagrangian with gauge symmetry in which, by dimensional reasons, an extra free parameter is introduced \cite{Bop}.  That free parameter is identified as the Podolsky’s mass, i.e., a mass associated with the Podolsky field.. This generalized electrodynamics gives the correct finite expressions for self-force of charged particles, as shown by Frenkel-Zayats \cite{Frank}, it gives rise to interesting effects produced by the presence of sources \cite{LW2}, and, as general rule, the theory offers better ultraviolet behavior (in a sense closely related to the Pauli-Villars-Rayski regularization scheme  regularization scheme \cite{PV, RThibes}). Although there is a characteristic length (the inverse of the Podolsky mass) associated with the size of the charged particles seen, for example, in the electron-positron scattering \cite{Daniel}, the possible source for the Podolsky’s mass, derived from the structure behind the self-interaction of the particles and their sizes, remains mysterious. But it is worth noticing that Podolsky electrodynamics breaks the dual symmetry \cite{Brandt} and lead us to speculate the existence of some mechanism that breaks the dual symmetry and generates mass.
From the point of view of degrees of freedom, generalized photons have five (Maxwell plus Proca) \cite{Anderson} and, in the free case, the generalized electrodynamics is equivalent to the Lee-Wick, wherein the gauge field is complex \cite{LeeW, FABAAN1}. These five degrees of freedom also have influence in the study of zero point energy \cite{Blaz, FabAnd}.

The generalized electrodynamics respect  some important general properties such as causality from the point of view of causal perturbation theory \cite{BufPimSot}, stability from the concept of Lagrangian anchor \cite{Russos}, and renormalizability as well \cite{Renorma}.

It is well known that the Stueckelberg field can be used to construct a gauge-invariant theory of first-order derivatives. One of the questions we intend to answer in this paper is whether a generalization of the Stueckelberg procedure can be devised for theories of second-order derivatives, in the Podolsky fashion. As we shall see, the answer is yes. Also, we shall show a connection between the generalized Stueckelberg method and a generalized Higgs mechanism. The importance of the analysis is a further study about quantization of gauge theories through constraint analysis and functional quantization \cite{Anderson,Galvao}. The motivation of the study is a better understanding of the properties of a Proca mass in generalized electrodynamics. As one knows, the way generalized photons gain mass has a phenomenological appeal not only in Higgs mechanism, but also in Debye shielding \cite{LW2,Deb}. A direct implication of the covariant quantization is the application of the fictitious parameters method in a new environment \cite{BRS-T,BRS-T2}, presenting an alternative of how some gauge parameters do not contribute to physics due to BRS-T symmetry.

In what follows, we will try to understand the mass generation mechanism in the generalized electrodynamics context, following the ideas in the previous works \cite{BRS-T, BRS-T2, Sto, Deb}. This work is organized as follows. In section \ref{classical} we define what we call the generalized Stueckelberg model \textcolor{blue}{,} discuss its classical aspects and construct the phase space for the model. In section \ref{quantum}, we construct the transition amplitude by the Fadeev-Senjanovic procedure and write it in the Fadeev-Popov covariant form. In Sec. \ref{BRS-T} we analyse the BRS-T symmetry and the invariance of the physics by gauge choices. In Sec. \ref{Higgs} we use the Stueckelberg structure to quantize the generalized Higgs gauge theory. Finally\textcolor{blue}{,} section \ref{conclusion} contains our conclusions and perspectives. The metric signature $(+,-,-,-)$ is used throughout.

\section{The Generalized Stueckelberg Gauge Theory}
\label{classical}

In the present section we will introduce the generalized Stueckelberg model and we will make a brief discussion of its canonical structure, including the constraint analysis, the canonical Hamiltonian, and its equations of motion.

We start with the following Lagrangian density

\begin{eqnarray}
\label{LS}
&&{\cal L}_S=-\frac{1}{4}F^{\mu\nu}F_{\mu\nu} +\frac{1}{2}a^{2}_{0}\partial_{\nu}F^{\nu\mu}\partial_{\rho}F^{\rho}_{\phantom{\rho}\mu} +\frac{1}{2}M^{2}S_{\mu}{S^{*}}^{\mu};\label{lagrangian}\\
&&S_{\mu}=A_{\mu}+\frac{1}{m_{s}}\partial_{\mu}B+\frac{i}{m_{p}^{2}}\partial_{\lambda}F^{\lambda}_{\phantom{\lambda}\mu},\label{S and A}
\end{eqnarray}
in which $F_{\mu\nu}=\partial_{\mu}A_{\nu}-\partial_{\nu}A_{\mu}$ are the components of the field-strength tensor,  $B$ is the Stueckelberg compensating field, and $a^{2}_{0}$, $M$, $m_s$, and $m_P$ are non-vanishing real parameters. The Lagrangian (\ref{lagrangian}) is invariant under the following $U\left(1\right)$ local gauge  transformation

\begin{eqnarray}
&&A_{\mu}\left(x\right)\rightarrow A_{\mu}\left(x\right)+\partial_{\mu}\alpha\left(x\right);\label{gsA}\\
&&B\left(x\right)\rightarrow B\left(x\right)-m_{s}\alpha\left(x\right),\label{gsB}
\end{eqnarray}
where $\alpha$ is a (smooth) Poincar{\'e} scalar field.

Our stepping stone Lagrangian density (\ref{lagrangian}) is to be interpreted as the usual Bopp-Podolsky Lagrangian density coupled with the Stueckelberg compensating field $B$ through (\ref{S and A}). Since our theory is originally one whose Lagrangian density depends on second-order derivatives of the gauge field $A$, we have added, for completeness, a term that contains second-order derivatives of $A$ to the usual interacting vector $S$  as well.

The previous Lagrangian (\ref{lagrangian}) can be  rewritten in a more convenient form as
\begin{equation}
\label{Lagra}
{\cal L}_S=-\frac{1}{4}F^{\mu\nu}F_{\mu\nu}+\frac{1}{2}a^{2}\partial_{\mu}F^{\mu\beta}\partial^{\nu}F_{\nu\beta} +\frac{1}{2}M^{2}\left(A_{\mu}+\frac{1}{m_{s}}\partial_{\mu}B\right)\left(A^{\mu}+\frac{1}{m_{s}}\partial^{\mu}B\right),
\end{equation}
where we have defined the (squared) changed Podolsky parameter $a^2={1}/{m^2}$ according to

\begin{equation}
 a^{2}\equiv a^{2}_{0}+\left(\frac{M}{m_{p}^{2}}\right)^{2}.\label{a0 to a}
\end{equation}
Although the translation $a_0^2\mapsto a^2$ acts as a simple relabeling, its purpose is making it clear that were we to turn off the interaction - that is, making $M=0$ in eq. (2.1), we would still get a second-order derivative theory for the free case. In this case the mechanisms studied generate a Proca mass and shift the Podolsky parameter, preserving the physical degrees of freedom. Otherwise, we would have a break in the conservation of physical degrees of freedom. The relevance of  defining this parameter is due to the fact that it controls the access to the $3$ phases of the system \cite{Deb}.

The Euler-Lagrange equations that follow from this Lagrangian density read

\begin{align}
(1+a^{2}\Box)(\eta_{\mu\nu}\Box-\partial_{\mu}\partial_{\nu})A^{\nu} +M^{2}\left(A_{\mu}+\frac{\partial_{\mu}B}{m_{s}}\right)=&\,0;\label{EL1}\\
\partial_{\mu}A^{\mu}+\frac{\Box B}{m_{s}}=&\,0,\label{EL2}
\end{align}
with explicit gauge symmetry, in view of the transformations (\ref{gsA}-\ref{gsB}). In order to solve the equations of motion, a gauge condition must be considered. Assuming this one
\begin{equation*} \Big(1+a^2\Box\Big)\partial^\mu A_\mu-\frac{M^2}{m_s}B=0 \end{equation*}
the equations of motion decouple and yield
 \begin{align}
\Big[\big(1+a^{2}\Box\big)\Box+M^2\Big]A_{\mu} =&\,0 \ ;\\
\Big[\big(1+a^{2}\Box\big)\Box+M^2\Big]B=&\,0.
\end{align}  
It is also possible to define the gauge invariant and transverse physical field
\begin{equation}
\label{trans}
\mathcal{U_\mu}=\Big(A_{\mu}+\frac{\partial_{\mu}B}{m_{s}}\Big)
\end{equation}
associated to the process of  ``eating the Goldstone boson"  in a generalized Higgs mechanism, as we will see posteriorly.

Now we will construct the phase space in the Hamilton formalism for the theory using the Dirac's methodology and the Ostrogradsky procedure, which is basically the Hamiltonian formalism for higher-order derivative theories. In this case, we have an extended phase space and, by applying the Dirac-Bergman algorithm, we will extract the constraints and degrees of freedom of the system. In order to write the Hamiltonian we will seek firstly the generator of the spacetime translations in the view of the Noether theorem, whose components are \cite{Noether, Barut}

\begin{equation}
\mathcal{T}^{\mu}_{\phantom{\mu}\lambda}=\frac{\partial{\cal L}}{\partial(\partial_{\mu}A_{\nu})}\partial_{\lambda}A_{\nu}-\partial_{\theta}\left[\frac{\partial{\cal L}}{\partial(\partial_{\mu}\partial_{\theta}A_{\nu})}\right]\partial_{\lambda}A_{\nu} +\frac{\partial{\cal L}}{\partial(\partial_{\mu}\partial_{\theta}A_{\nu})}\partial_{\theta}\partial_{\lambda}A_{\nu}+\frac{\partial{\cal L}}{\partial(\partial_{\mu}B)}\partial_{\lambda}B-\eta^{\mu}_{\lambda}{\cal L}.
\end{equation}

The canonical Hamiltonian $H_{c}=\int d^{3}x\, \mathcal{T}^{0}_{\phantom{0}0}\equiv \int d^{3}x\mathcal{H}_c$ is written explicitly as

\begin{equation}
\label{HS}
H_{c}=\int d^{3}x\left(\Pi^{\mu}\partial_{0}A_{\mu}+\Phi^{\mu}\partial_{0}\Gamma_{\mu}+\varPi\partial_{0}B-{\cal L}\right),
\end{equation}
where $\Gamma_{\mu}\equiv\partial_{0}A_{\mu}$ and with the following Ostrogradsky canonical momenta associated with the fields $A_\mu$, $\Gamma_\mu$, and $B$, respectively \cite{Ostrogradsky}. The expression of the Hamiltonian suggests a generalized structure for the phase space associated to the fact that the Lagrangian depends on the second derivatives of the fields. The extended phase space has the following non-vanishing fundamental Poisson Brackets
\begin{eqnarray}
\left\{A^{\mu}\left(\mathbf{x},t\right),\Pi_{\nu}\left(\mathbf{y},t\right)\right\}_{P}&=&\,\delta^{\mu}_{\nu}\delta^{(3)}(\mathbf{x}-\mathbf{y});\\
\left\{\Gamma^{\mu}\left(\mathbf{x},t\right),\Phi_{\nu}\left(\mathbf{y},t\right)\right\}_{P}&=&\,\delta^{\mu}_{\nu}\delta^{(3)}(\mathbf{x}-\mathbf{y});\\
\left\{B\left(\mathbf{x},t\right),\varPi\left(\mathbf{y},t\right)\right\}_{P}&=&\,\delta^{(3)}(\mathbf{x}-\mathbf{y}).
\end{eqnarray}
which characterizes the Ostrogadsky structure of the model.\\
\indent The canonical momenta are given below

\begin{eqnarray}
\Pi^{\mu}&=&\,F^{\mu 0}+\frac{1}{m_s^{2}}\left(\eta^{\mu j}\partial_{j}\partial_{l}F^{l0} -\partial_{0}\partial_{\nu}F^{\nu\mu}\right);\\
\Phi^{\mu}&=&\,\frac{1}{m_s^{2}}\left(\partial_{\nu}F^{\nu\mu}-\eta^{\mu 0}\partial_{j}F^{j0}\right);\\
\varPi&=&\,\frac{M^{2}}{m_{s}}\left(A_{0}+\frac{1}{m_{s}}\partial_{0}B\right).
\end{eqnarray}

Some of these momenta furnish the primary constraints of the theory\footnote{Throughout this work the symbol $\approx$ is used to denote \textit{weak equality} \cite{Dirac}.}

\begin{eqnarray}
\varphi_{1}\left(x\right)&=&\,\Phi_{0}\left(x\right)\approx 0;\label{constraint 1}\\
\varphi_{2}\left(x\right)&=&\,\Pi_{0}\left(x\right)-\overrightarrow{\partial}\cdot\overrightarrow{\Phi}\left(x\right)\approx 0,\label{constraint 2}
\end{eqnarray}

Now, we are able to write the canonical Hamiltonian density (\ref{HS}) in its explicit form as
\begin{align}
\mathcal{H}_C=&\,\Pi^0\Gamma^0+\vec \Phi.\vec{\partial}\Gamma^0-\overrightarrow{\Pi}\cdot\overrightarrow{\Gamma}-\frac{m^2}{2}\left|\overrightarrow{\Phi}\right|^2+\Phi^k\partial_jF^{jk} -\frac{1}{2}\left|\overrightarrow{\Gamma}+\overrightarrow{\partial}A^0\right|^2  -\frac{1}{2m^2}\left(\overrightarrow{\partial}^2A^0+\overrightarrow{\partial}\cdot\overrightarrow{\Gamma}\right)^2 +\nonumber\\
&\,+\frac{1}{4}F^{jk}F_{jk}+\frac{m^2_s}{2M^2}\varPi^2-m_sA^0\varPi+\frac{M^2}{2}\left|\overrightarrow{A}-\frac{1}{m_s}\overrightarrow{\partial}B\right|^2.\label{H C}
\end{align}

The Ostrogradsky analysis leads us to the following non-vanishing fundamental Poisson brackets:

\begin{eqnarray}
\left\{A^{\mu}\left(\mathbf{x},t\right),\Pi_{\nu}\left(\mathbf{y},t\right)\right\}_{P}&=&\,\delta^{\mu}_{\nu}\delta^{(3)}(\mathbf{x}-\mathbf{y});\\
\left\{\Gamma^{\mu}\left(\mathbf{x},t\right),\Phi_{\nu}\left(\mathbf{y},t\right)\right\}_{P}&=&\,\delta^{\mu}_{\nu}\delta^{(3)}(\mathbf{x}-\mathbf{y});\\
\left\{B\left(\mathbf{x},t\right),\varPi\left(\mathbf{y},t\right)\right\}_{P}&=&\,\delta^{(3)}(\mathbf{x}-\mathbf{y}).
\end{eqnarray}

Now, following the Dirac methodology, we define the primary Hamiltonian density ${\cal H}_{p}={\cal H}_{c}+C_{1}\varphi^{1}+C_{2}\varphi^{2}$, where $\varphi^1$ and $\varphi^2$ are the constraints (\ref{constraint 1}) and (\ref{constraint 2}) and $C_1$ and $C_2$ are the respective Lagrange multipliers. The consistence conditions for each constraint, namely,

\begin{equation}
 \partial_{0}\varphi_{j}\left(\mathbf{x},t\right)=\int d^{3}{y}\{\varphi_{j}\left(\mathbf{x},t\right),{\cal H}_{p}\left(\mathbf{y},t\right)\}_{P}\approx 0,
\end{equation}
gives us only one additional constraint: the Gauss' law
\begin{equation}
 \varphi_{3}\left(x\right)\equiv m_{s}\varPi\left(x\right)+\overrightarrow{\partial}\cdot\overrightarrow{\Pi}\left(x\right)\approx 0,\label{constraint 3}
\end{equation}
whose consistency condition is tautological. As such, we have the full set of constraints $\left\{\varphi_{1},\varphi_{2},\varphi_{3}\right\}$ and the extended Hamiltonian density ${\cal H}_{E}={\cal H}_{p}+C_{3}\varphi^{3}$ with a new Lagrange multiplier $C_3$. Moreover, this set of constraints is first class, which means $\left\{\varphi_{j}\left(\mathbf{x},t\right),\varphi_{k}\left(\mathbf{y},t\right)\right\}_P = 0$ for all $j,k\in\left\{1,2,3\right\}$.

Using the extended Hamiltonian as the generator of the time evolution, we find the Hamilton equations\footnote{For any field $\phi\left(\mathbf{x},t\right)$ we have $\dot{\phi}\left(\mathbf{x},t\right)=\left\{\phi\left(\mathbf{x},t\right),\int d^3y \mathcal{H}_e\left(\mathbf{y},t\right)\right\}_P$.}

\begin{align}
\dot{A}^\mu\left(x\right)=&\,\delta^\mu_k\left[\Gamma^k\left(x\right)-\partial_kC_3\left(x\right)\right]+\delta_\mu^0\Gamma^0+\delta^\mu_0C_2\left(x\right);\label{dot A mu}\\
\dot{\Gamma}^\mu\left(x\right)=&\,\delta_\mu^i\partial_i\Gamma^0-\delta^\mu_k\left[m^2\Phi_k\left(x\right)+\partial_jF^{jk}\left(x\right)-\partial_kC_2\left(x\right)\right] +\delta^\mu_0C_1\left(x\right);\label{dot Gamma mu}\\
\dot{B}\left(x\right)=&\,\frac{m_s^2}{M^2}\varPi\left(x\right)-m_s A^0\left(x\right)+m_sC_3\left(x\right);\label{dot B}\\
\dot{\Pi}_\mu\left(x\right)=&\,-\delta^k_\mu\overrightarrow{\partial}^2\Phi_k\left(x\right) -\partial^j_\mu\overrightarrow{\partial}\cdot\overrightarrow{\Phi}\left(x\right)
-\delta^0_\mu\left\{\overrightarrow{\partial}\cdot\overrightarrow{\Gamma}\left(x\right)+\overrightarrow{\partial}^2A_0\left(x\right)\right.\nonumber\\ &\,\left.-\frac{1}{m^2}\overrightarrow{\partial}^2\left[\overrightarrow{\partial}^2A^0\left(x\right) +\overrightarrow{\partial}\cdot\overrightarrow{\Gamma}\left(x\right)\right]\right\} -\frac{1}{2}\left(\delta^j_\mu\partial_k-\delta^k_\mu\partial_j\right)F^{kj}\left(x\right)\nonumber\\
&\,+\delta^0_\mu m_s\varPi\left(x\right)-M^2\delta^k_\mu\left[A^k-\frac{1}{m_s}\partial_k B\left(x\right)\right];\\
\dot{\Phi}_\mu\left(x\right)=&\,\delta^k_\mu\left\{\Gamma^k\left(x\right)+\partial_k A^0\left(x\right)-\Pi_k\left(x\right) -\frac{1}{m^2}\partial_k\left[\overrightarrow{\partial}^2A^0\left(x\right)+\overrightarrow{\partial}\cdot\overrightarrow{\Gamma}\left(x\right)\right]\right\}+\delta_\mu^0(\Pi_0-\vec{\partial}.\vec \Phi);\\
\dot{\varPi}\left(x\right)=&\,-\frac{M^2}{m_s}\left[\overrightarrow{\partial}\cdot\overrightarrow{A}\left(x\right) -\frac{1}{m_s}\overrightarrow{\partial}^2 B\left(x\right)\right].\label{dot varPi}
\end{align}

As we can see, there are some arbitrariness in the Hamilton equations due to the Lagrange multiplier dependence. Dirac conjectured that the first-class constraints are the generators of the gauge transformations \cite{Henneaux Teitelboim}. In order to set this arbitrariness aside, we must impose (three) extra constraints on the system in such a way that the new set of constraints $\Omega=\left\{\omega_1,\omega_2,\omega_3,\omega_4,\omega_5,\omega_6\right\}$ is second class. This means that for any $j\in\left\{1,2,3,4,5,6\right\}$ there is at least one $k\in\left\{1,2,3,4,5,6\right\}$ for which $\left\{\omega_{j}\left(\mathbf{x},t\right),\omega_{k}\left(\mathbf{y},t\right)\right\}_P \neq 0$. As a result,  the Hamilton equations are dynamically consistent with the Euler-Lagrange equations. In order to do that, we firstly notice that, thanks to (\ref{dot varPi}), the time derivative of equation (\ref{dot B}) leads to

\begin{equation}
\partial_\mu A^\mu\left(x\right)+\frac{\Box}{m_s}B\left(x\right)=\frac{\partial C_3\left(x\right)}{\partial t}.\label{finding C3}
\end{equation}

Comparing this equation with (\ref{EL2}) tells us that the Lagrange multiplier $C_3$ must be time independent. However, the left-hand-side of (\ref{finding C3}) is a Lorentz scalar, which means its right-hand-side must be invariant under Lorentz as well. The only way a Lorentz scalar is time-independent is if it is event-independent. Therefore, the Lagrange multiplier $C_3$ is a constant which, without loss of generality, may be taken to be zero.

Before we go into details as how to find the other unknown functions, whose computations are more involved, let us notice that the zero-th components of equations (\ref{dot A mu}) and (\ref{dot Gamma mu}) read

\begin{align}
C_1\left(x\right)=&\,\dot{\Gamma}^0\left(x\right);\label{C1}\\
C_2\left(x\right)=&\,\dot{A}^0\left(x\right)-\Gamma^0.\label{C2}
\end{align}

The main goal of the present procedure is obtaining Hamilton equations which are equivalent to the Euler-Lagrange equations. We already used one of them to discover the Lagrange multiplier $C_3$. In order to find the other functions, we shall use the other Euler-Lagrange equations. To be more precise, let us look into the zero-th component of (\ref{EL1}):

\begin{align}
\left[M^2-\left(1+\frac{\Box}{m^2}\right)\overrightarrow{\partial}^2\right]A^0\left(x\right) =\left(1+\frac{\Box}{m^2}\right)\frac{\partial\left(\overrightarrow{\partial}\cdot\mathbf{A}\left(x\right)\right)}{\partial t}-\frac{M^2}{m_s}\frac{\partial B\left(x\right)}{\partial t}.
\end{align}

This equation can be solved \textit{formally} as

\begin{align}
A^0\left(x\right) =\left[M^2-\left(1+\frac{\Box}{m^2}\right)\overrightarrow{\partial}^2\right]^{-1}\frac{\partial}{\partial t} \left[\left(1+\frac{\Box}{m^2}\right)\overrightarrow{\partial}\cdot\mathbf{A}\left(x\right)-\frac{M^2}{m_s} B\left(x\right)\right].
\end{align}

In complete analogy with \cite{Galvao}, we choose the \textit{modified generalized Coulomb condition}

\begin{equation}
\left(1+\frac{\Box}{m^2}\right)\overrightarrow{\partial}\cdot\mathbf{A}\left(x\right)\approx\frac{M^2}{m_s} B\left(x\right).\label{modified coulomb}
\end{equation}

This condition and its consistency  - that is, the fact that it must hold for all time - imply $A^0\left(x\right)\approx0$ and $\dot{A}^0\left(x\right)=\Gamma^0\left(x\right)\approx0$. These equalitlies, through relations (\ref{C1}) and (\ref{C2}), define the last two unknown Lagrange multiplier to be both identically equal to zero as well.

So, we are lead to choose the following gauge conditions:

\begin{align}
\Sigma_1\left(x\right)\equiv&\,\Gamma^0\left(x\right)\approx 0;\label{gauge 1}\\
\Sigma_2\left(x\right)\equiv&\, A^0\left(x\right)\approx 0;\label{gauge 2}\\
\Sigma_3\left(x\right)\equiv&\,\left(1+\frac{\Box}{m^2}\right)\overrightarrow{\partial}\cdot\mathbf{A}\left(x\right)-\frac{M^2}{m_s} B\left(x\right)\approx 0.\label{gauge 3}
\end{align}

Taking into account the first-class constraints (\ref{constraint 1}), (\ref{constraint 2}), and \ref{constraint 3}), as well as the gauge conditions (\ref{gauge 1}), (\ref{gauge 2}), and (\ref{gauge 3}), we define:

\begin{eqnarray}
  \omega_j\left(x\right) &=& \left\{\begin{array}{ll}
                                      \varphi_j\left(x\right), & \mbox{for }j\in\left\{1,2,3\right\}; \\
                                      \Sigma_{j-3}\left(x\right) & \mbox{for }j\in\left\{4,5,6\right\}.
                                    \end{array}\right.
\end{eqnarray}

Therefore, we achieved our goal that was to find a set of constraints that is of second class: $\Omega=\left\{\omega_1,\omega_2,\omega_3,\omega_4,\omega_5,\omega_6\right\}$, as evidenced by the non-vanishing determinant of the matrix $\left[\Omega_{jk}\right]$, whose $jk$-th component is $\Omega_{jk}\equiv\left\{\omega_j\left(\mathbf{x},t\right),\omega_k\left(\mathbf{y},t\right)\right\}_P$:

\begin{eqnarray}
  \left[\Omega_{jk}\right] &=& \left[\begin{array}{cccccc}
                                 0 & 0 & 0 & -1 & 0 & 0 \\
                                 0 & 0 & 0 & 0 & -1 & 0 \\
                                 0 & 0 & 0 & 0 & 0 & \left(1+\frac{\Box}{m^2}\right)\overrightarrow{\partial}^2 -M^2 \\
                                 1 & 0 & 0 & 0 & 0 & 0 \\
                                 0 & 1 & 0 & 0 & 0 & 0 \\
                                 0 & 0 & -\left(1+\frac{\Box}{m^2}\right)\overrightarrow{\partial}^2+M^2 & 0 & 0 & 0
                               \end{array}\right]\delta^{\left(3\right)}\left(\mathbf{x}-\mathbf{y}\right).\nonumber\\
\end{eqnarray}

Before proceeding to the quantization of the theory, let us make a quick counting of the degrees of freedom. We have a set of $18$ variables $\left\{A^{\mu}\left(x\right),\Gamma^{\mu}\left(x\right),B,\Pi_{\nu}\left(x\right),\Phi_{\mu}\left(x\right),\varPi\left(x\right)\right\}$ and $6$ constraints - the $\Omega$ set - to describe the phase space. So, we ended up with $12$ independent fields in the phase space or, alternatively, $6$ degrees of freedom in the configuration space. This is the same number of degrees of freedom a theory comprising the Podolsky field with a real scalar field has $\left(5+1\right)$ which, in its turn, is the same one as Maxwell, Proca, and Scalar $(2+3+1)$ \cite{Bonin}. This makes it clear that the ultimate effect of the Proca mass term in (\ref{lagrangian}) is to add one degree of freedom to the already existing $5$ of the usual free Podolsky theory. This goes against the naive notion that the Lagrangian density (\ref{Lagra}) might have given that the final theory would have $7$ degrees of freedom - that is, two massive photons and the scalar Stueckelberg field. Our results show that this does not happen  because the gauge symmetry shields this degree of freedom.

It is interesting to mention that the Podolsky-Stueckelberg model has several interesting properties that motivates its study. For the case of the Podolsky's $\textit{QED}$ plasma (massless and massive photons interacting with fermions in finite temperature), a temperature dependent extra term arises in the propagator. Then, it acquires two massive poles. The low energy structure of this new contribution can be modelled by adding a Proca mass term to the Podolsky model. We verified that the structure of the resulting model implies in the possibility of three different phases associated to the behaviour of the non-relativistic potential. They are accessed through specific relations between the values of the Proca and Podolsky masses \cite{Deb}. This interaction can be expressed as the superposition of two Yukawa potentials, an exponentially decaying function, or even present an oscillatory structure. This last property occurs when the poles become complex  \footnote{It can be achieved by a specific relation between the Proca and Podolsky masses.} \cite{Gribov}. The fact that the structure of the physical part of the two-point function is proportional to the difference of two free massive propagators leads to an ultraviolet improvement associated to a well behaved potential at the origin. There is also another notable property related to the one-loop correction to the vertex function for a model composed by this Podolsky-Proca gauge field in interaction with fermions. Comparing to the case of $\textit{QED}$, the infrared and ultraviolet divergent terms that arise, in its low energy limit, cancel out with the remaining part being proportional to the finite term $\ln{m_+/m_-}$ in which $m_\pm$ denotes the pole masses \cite{Schwartz} of the gauge field propagator.

\section{The (quantum) transition amplitude}\label{quantum}

Since we have constructed the phase space of the model following Dirac methodology, in this section we will write down the transition amplitude of the quantum theory in a covariant fashion \cite{Constraint,Zambrano}.

With the full set of constraints, the next step is to obtain the transition amplitude for the quantum theory. It is originally written as

\begin{align}
Z=\int \left[\prod_{\mu=0}^3\mathcal{D}A^\mu\mathcal{D}\Pi_\mu\mathcal{D}\Gamma^\mu\mathcal{D}\Phi_\mu\right]\mathcal{D}B\mathcal{D}\varPi \sqrt{\det{\left[\Omega_{jk}\right]}}\left[\prod_{l=1}^6\delta\left(\omega_l\right)\right]e^{i\int d^4x\mathcal{L}_C}
\end{align}
where $\mathcal{L}_C$ is the Lagrangian density in the canonical form, defined in terms of (\ref{H C}) as

\begin{align}
\mathcal{L}_C=\Pi_\mu \dot{A}^\mu +\Phi_\mu\dot{\Gamma}^\mu+\varPi\dot{B}-\mathcal{H}_C.
\end{align}

After integrating over some variables\textcolor{blue}{,} we are left with the following transition amplitude:

\begin{align}
Z=\int \left[\prod_{\mu=0}^3\mathcal{D}A^\mu\right]\mathcal{D}B\det\left(\left(1+\frac{\Box}{m^2}\right)\overrightarrow{\partial}^2-M^2\right) \delta\left(\left(1+\frac{\Box}{m^2}\right)\overrightarrow{\partial}\cdot \mathbf{A}-\frac{M^2}{m_s}B\right)e^{i\int d^4x\mathcal{L}_S},\label{Z in a frame}
\end{align}
where $\mathcal{L}_S$ is the Lagrangian density (\ref{Lagra}).

Although the transition amplitude above  is correct, it is not Lorentz covariant. This is due to the fact that it is written in the modified generalized Coulomb gauge (\ref{modified coulomb}). Such a gauge is valid in a particular inertial frame. In order to write $Z$ is a covariant form, that is, a form that remains the same in every inertial frame, we define the \textit{modified generalized Lorenz gauge} as

\begin{align}
\left(1+\frac{\Box}{m^2}\right)\partial_\mu A^\mu\left(x\right) = g\frac{M^2}{m_s}B\left(x\right)+f\left(x\right)\label{modified lorenz}
\end{align}
where $g$ is a dimensionless real number and $f$ an arbitrary Lorentz scalar function.

Using the so-called Fadeev-Popov trick, we can rewrite (\ref{Z in a frame}) as

\begin{align}
Z=\int \left[\prod_{\mu=0}^3\mathcal{D}A^\mu\right]\mathcal{D}B\det\left(\left(1+\frac{\Box}{m^2}\right)\Box+gM^2\right) \delta\left(\left(1+\frac{\Box}{m^2}\right)\partial_\nu A^\nu-g\frac{M^2}{m_s}B-f\right)e^{i\int d^4x\mathcal{L}_S}.\label{Z in any frame}
\end{align}

Now, we define a measure $\mathcal{M}$ which is a functional of the scalar function $f$:

\begin{align}
\int d\mathcal{M}\left[f\right]=1.
\end{align}

As shown in \cite{thooft}, an adequate choice for such a measure is $d\mathcal{M}\left[f\right]=\mathcal{D}f$ \\ $\exp\left\{-i\int d^4x\left[f\left(x\right)\right]^2/\left[2\left(\xi+0^+\right)\right]+\ln\left(-2i\pi\xi\right)/2\right\}$. With this, up to a (possibly ill-defined) normmalization constant, the transition amplitude assumes the simple, covariant form:

\begin{align}
Z=\int \left[\prod_{\mu=0}^3\mathcal{D}A^\mu\right]\mathcal{D}B\mathcal{D}\bar{c}\mathcal{D}c\, e^{i\int d^4x\mathcal{L}_{{eff}}},\label{transition}
\end{align}
where the effective Lagrangian density $\mathcal{L}_{{eff}}$ is defined as

\begin{align}
\label{tampli}
\mathcal{L}_{{eff}}\equiv \mathcal{L}-\frac{1}{2\xi}\left[\left(1+\frac{\Box}{m^2}\right)\partial_\mu A^\mu-g\frac{M^2}{m_s}B\right]^2 +\bar{c}\left[\left(1+\frac{\Box}{m^2}\right)\Box+gM ^2\right]c,
\end{align}
where $\bar{c}$ and $c$ are Grassmannian scalar fields.

Next, we shall study the BRS-T symmetry of this theory.

\section{BRS-T symmetry and fictitious parameters}
\label{BRS-T}

As we are aware, in the covariant transition amplitude the Lorenz gauge condition is generalized to the $R_{\xi}$ gauges [$\frac{1}{2\xi}(\partial_{\mu}A^{\mu})^{2}$] and more than that, we have the 't Hooft-Veltman gauges [$\frac{1}{2\xi}(\partial_{\mu}A^{\mu}+gA_{\mu}A^{\mu})^{2}$] \cite{Nash}, in which  we have a fictitious interaction between the ghosts and the gauge field, characterized by the $g$ parameter. Surprisingly, the transition amplitude does not depend on this $g$ parameter due to the BRS-T symmetry \cite{BRS-T, BRS-T2}, so this is a dummy parameter and we can take $g=0$. The previous statement is related to the fact that physics (S-matrix) does not depend on $g$ \cite{Smatrix} when we include the fermions and study the effective action. The same line of reasoning can be applied in generalized electrodynamics (Podolsky), wherein we have the generalized covariant gauge choices and a dummy parameter. The independence of the amplitudes from the parameter $g$ is to be expected since the physical field $\mathcal{U}_\mu$ discussed previously in eq. (\ref{trans}) has equations of motion that do not depend on any gauge fixation.

Now the question that arises is whether the same methodology of the fictitious parameters applies to the 't Hooft gauge [$\frac{1}{2\xi}(\partial_{\mu}A^{\mu}+gqv\zeta)^{2}$] \cite{Framp} when we study the Higgs model and the spontaneous symmetry breaking. Here $gqv$ represents a fictitious mass to the Goldstone boson $\zeta$. As the transition amplitude for the generalized Stueckelberg theory in eq. (\ref{tampli}) has the same 't Hooft gauge structure in the limit of $m\rightarrow\infty$, with the Stueckelberg field $B$ playing the role of the Goldstone boson, let us take a look on the parameter $g$ and see whether it is fictitious.

In other words, the purpose of this section is to learn whether or not the parameter $g$ appearing for the first time in the modified generalized Lorenz condition (\ref{modified lorenz}) and used explicitly in the transition amplitude (\ref{transition}) \textit{via} (\ref{tampli}) has a physical meaning.

Before we proceed, let us define the generating functional\footnote{As usual, $Z_{GF}\left[0,0,0,0\right]=Z$, the transition amplitude.}

\begin{align}
Z_{GF}\left[J,j,\bar{\zeta},\zeta\right]=N\int \left[\prod_{\rho=0}^3\mathcal{D}A^\rho\right]\mathcal{D}B\mathcal{D}\bar{c}\mathcal{D}c e^{iS\left[A,B,c,\bar{c};J,j,\bar{\zeta},\zeta\right]}\label{generating functional}
\end{align}
where $J$, $j$, $\bar{\zeta}$, and $\zeta$ are classical sources (the last two being also Grassmannian) and

\begin{align}
S\left[A,B,c,\bar{c};J,j,\bar{\zeta},\zeta\right]\equiv\int d^4x\mathcal{L}_{eff}+\int d^4x\left[J_\mu\left(x\right)A^\mu\left(x\right) +j\left(x\right)B\left(x\right)+\bar{\zeta}\left(x\right)c\left(x\right)+\zeta\left(x\right)\bar{c}\left(x\right)\right].
\end{align}

We also define the \textit{BRS-T transformation} as

\begin{align}
A^\mu\left(x\right)\rightarrow&\,\tilde{A}^\mu\left(x\right)\equiv A^\mu\left(x\right)+\delta A^\mu\left(x\right);\label{transformation A}\\
B\left(x\right)\rightarrow&\,\tilde{B}\left(x\right)\equiv B\left(x\right)+\delta B\left(x\right);\\
c\left(x\right)\rightarrow&\,\tilde{c}\left(x\right)\equiv c\left(x\right)+\delta c\left(x\right);\\
\bar{c}\left(x\right)\rightarrow&\,\tilde{\bar{c}}\left(x\right)\equiv\bar{c}\left(x\right)+\delta \bar{c}\left(x\right),\label{transformation c bar}
\end{align}
with
\begin{align}
\delta A^\mu\left(x\right)\equiv&\,\lambda\partial^\mu c\left(x\right);\\
\delta B\left(x\right)\equiv&\,\lambda m_s c\left(x\right);\\
\delta c\left(x\right)\equiv&\,0;\\
\delta \bar{c}\left(x\right)\equiv&\,\frac{\lambda}{\xi}\left[\left(1+\frac{\Box}{m^2}\right)\partial_\mu A^\mu\left(x\right) -g\frac{M^2}{m_s}B\left(x\right)\right].\label{BRS-T c bar}
\end{align}

Here, $\lambda$ is a non-vanishing Grassmannian constant.

From now on, we shall write

\begin{align}
\delta_{BRS-T}Q\equiv \tilde{Q}-Q
\end{align}
for any quantity $Q$. The symbol $\tilde{Q}$ signifies the quantity $Q$ transformed under (\ref{transformation A}-\ref{BRS-T c bar}).

\begin{lemma}
  The generating functional (\ref{generating functional}) is invariant under the transformation (\ref{transformation A}-\ref{transformation c bar}).\label{lemma GF}
\end{lemma}
\begin{proof}
  As a matter of fact, the transformed generating functional is

  \begin{align}
  \tilde{Z}_{GF}\left[J,j,\bar{\zeta},\zeta\right]=&\,N\int \left[\prod_{\rho=0}^3\mathcal{D}\tilde{A}^\rho\right]\mathcal{D}\tilde{B} \mathcal{D}\tilde{\bar{c}}\mathcal{D}\tilde{c} e^{iS\left[\tilde{A},\tilde{B},\tilde{c},\tilde{\bar{c}};J,j,\bar{\zeta},\zeta\right]}\\
  =&\,N\int \left[\prod_{\rho=0}^3\mathcal{D}A^\rho\right]\mathcal{D}B\mathcal{D}\bar{c}\mathcal{D}c e^{iS\left[A,B,c,\bar{c};J,j,\bar{\zeta},\zeta\right]}\\
  =&\,Z_{GF}\left[J,j,\bar{\zeta},\zeta\right],
  \end{align}
  where we have used the fact that the integration ``variables'' are dummies.
\end{proof}

\begin{lemma}
The integration measure $\left[\prod_{\rho=0}^3\mathcal{D}A^\rho\right]\mathcal{D}B\mathcal{D}\bar{c}\mathcal{D}c $ is invariant under the BRS-T transformation (\ref{transformation A}-\ref{BRS-T c bar}).\label{lemma measure}
\end{lemma}

\begin{proof}
The integration measure $\left[\prod_{\rho=0}^3\mathcal{D}A^\rho\right]\mathcal{D}B\mathcal{D}\bar{c}\mathcal{D}c $ transforms under (\ref{transformation A}-\ref{BRS-T c bar}) as

\begin{align}
\left[\prod_{\rho=0}^3\mathcal{D}\tilde{A}^\rho\right]\mathcal{D}\tilde{B} \mathcal{D}\tilde{\bar{c}}\mathcal{D}\tilde{c} = J\left[\prod_{\rho=0}^3\mathcal{D}A^\rho\right]\mathcal{D}B\mathcal{D}\bar{c}\mathcal{D}c
\end{align}
where $J$ is the Jacobian of the transformation which, in simplified notation, reads
\begin{eqnarray}
J&=&\,\left|\det\left[
\begin{array}{cccc}
  \delta^\mu_\nu & 0 & \frac{\lambda}{\xi}\left(1+\frac{\Box}{m^2}\right)\partial_\nu & 0 \\
  0 & 1 & -\frac{\lambda}{\xi}g\frac{M^2}{m_s} & 0 \\
  0 & 0 & 1 & 0 \\
  -\lambda\partial^\mu & -\lambda m_s & 0 & 1
\end{array}
\right]\delta^{\left(4\right)}\left(x-y\right)\right|\nonumber\\
&=&\delta^\mu_\nu\delta^{\left(4\right)}\left(x-y\right).
\end{eqnarray}
Therefore,
\begin{align}
\delta_{BRS-T}\left(\left[\prod_{\rho=0}^3\mathcal{D}A^\rho\right]\mathcal{D}B\mathcal{D}\bar{c}\mathcal{D}c\right)=0.
\end{align}
\end{proof}

We shall show the parameter $g$ is fictitious through the following

\begin{theorem}
 The transition amplitude $Z$ (\ref{transition}) is independent of the parameter $g$.\label{theorem}
\end{theorem}
\begin{proof}
  The derivative of the transition amplitude (\ref{transition}) with respect to $g$ reads\footnote{We used the widespread notation $\left\langle O\right\rangle=\int \left[\prod_{\mu=0}^3\mathcal{D}A^\mu\right]\mathcal{D}B\mathcal{D}\bar{c}\mathcal{D}c\, O e^{i\int d^4x\mathcal{L}_{{eff}}}$ and
   $\left\langle O\right\rangle_s=\int \left[\prod_{\mu=0}^3\mathcal{D}A^\mu\right]\mathcal{D}B\mathcal{D}\bar{c}\mathcal{D}c\, O e^{iS\left[A,B,c,\bar{c};J,j,\bar{\zeta},\zeta\right]}$ for any operator $O$.}

  \begin{align}
    \frac{\partial Z}{\partial g}=i\int d^4x\left\langle \frac{1}{\xi}\left[\left(1+\frac{\Box}{m^2}\right)\partial_\mu A^\mu\left(x\right) -g\frac{M^2}{m_s}B\left(x\right)\right]\frac{M^2}{m_s}B\left(x\right)+M^2\bar{c}\left(x\right)c\left(x\right) \right\rangle.\label{partial Z partial g}
  \end{align}
  Computing $\delta_{BRS-T}Z_{GF}\left[J,j,\bar{\zeta},\zeta\right]$ and taking into account the lemmas (\ref{lemma GF}) and (\ref{lemma measure}) furnishes

  \begin{align}
  i\lambda\int d^4z\left\langle J_\mu\left(z\right)\partial^\mu c\left(z\right)+j\left(z\right)m_sc\left(z\right) +\frac{\zeta\left(z\right)}{\xi}\left[\left(1+\frac{\Box}{m^2}\right)\partial_\mu A^\mu\left(z\right) - g\frac{M^2}{m_s}B\left(z\right)\right]\right\rangle_s=0.
  \end{align}

  Applying the functional differential operator ${\delta^2}/{\delta j\left(x\right)\delta \zeta\left(x\right)}$ to both sides of this identity and subsequently making the classical sources vanish yields

  \begin{align}
  \lambda\left\langle\frac{1}{\xi}\left[\left(1+\frac{\Box}{m^2}\right)\partial_\mu A^\mu\left(y\right) - g\frac{M^2}{m_s}B\left(y\right)\right]B\left(x\right)+m_s\bar{c}\left(y\right)c\left(x\right)\right\rangle=0.
  \end{align}
  Taking into account that $\lambda$ is not zero, choosing $y=x$, and inserting the result into (\ref{partial Z partial g}) shows us that

  \begin{align}
  \frac{\partial Z}{\partial g}=0,
  \end{align}
  or, in other words, the transition amplitude $Z$ (\ref{transition}) is independent of the parameter $g$.
\end{proof}

Since the transition amplitude remains the same for any parameter $g$ and all physical quantities depend upon the transition amplitude, we conclude that no physical quantity depends on $g$. In other words, $g$ is a \textit{fictitious parameter}. We usually choose $g=\xi$ when we deal with t'Hooft gauges \cite{thooft, Framp}. Finally, it is important to see that when we choose $g=\xi$ the gauge field A and the auxiliary B field do not, in general, decouple, just in the limit $m_{P}\rightarrow\infty$. To have the decoupling of the Stueckelberg-photon interaction in higher derivative theory we need to generalize the Stueckelberg Lagrangian term $\frac{1}{2}M^{2}S_{\mu}{S^{*}}^{\mu}$ in the following form

\begin{equation}
\label{decoupled}
S_{\mu}=A_{\mu}+\frac{1}{m_{s}}(1+\frac{\Box}{m_{P}^{2}})\partial_{\mu}B+\frac{i}{m_{p}^{2}}\partial_{\lambda}F^{\lambda}_{\phantom{\lambda}\mu}.
\end{equation}
So if we choose $g=\xi$ we decouple the fields and have massive Stueckelberg and ghost fields. On the other hand, if we choose $g=0$ we couple the fields and we have non massive Stueckelberg and ghost fields, independent of the limit $m_{P}\rightarrow\infty$.

\section{Generalized Higgs gauge theory}
\label{Higgs}

What we call the generalized Higgs model is defined  by the Lagrangian density

\begin{align}
\label{Laghiggs}
\mathcal{L}_H=&\, -\frac{1}{4}F^{\mu\nu}F_{\mu\nu}+\frac{1}{2}a_0^2\partial_\mu F^{\mu\nu}\partial_\rho F^{\rho}_{\phantom{\rho}\nu} +\nabla_\mu\phi\left(\nabla^\mu\phi\right)^*+\mu_s \left|\phi\right|^2-\lambda_s\left|\phi\right|^4,
\end{align}
where the first two terms are like those of (\ref{LS}), $\phi$ is a complex scalar field, $\mu_s $ and $\lambda_s$ are real parameters with $\lambda_s>0$, and

\begin{align}
\nabla_{\mu}\equiv&\, \partial_\mu -iq\left(A_\mu+iZ_\mu\right)
\end{align}
with $Z_\mu\equiv\partial^\nu F_{\nu\mu}/m_P^2$ is an extension of the minimal coupling.

First of all, we notice that this model is invariant under a $U\left(1\right)$-gauge symmetry:

\begin{align}
\phi\left(x\right)\rightarrow&\,e^{i\alpha\left(x\right)}\phi\left(x\right);\\
\phi^*\left(x\right)\rightarrow&\,e^{-i\alpha\left(x\right)}\phi^*\left(x\right);\\
A_\mu\left(x\right)\rightarrow&\,A_\mu\left(x\right)+\frac{1}{q}\partial_\mu\alpha\left(x\right),
\end{align}
$\alpha$ being a (sufficiently smooth) scalar field.

Up to a total-derivative term, (\ref{Laghiggs}) can be brought to the form

\begin{align}
\mathcal{L}_H=&\, -\frac{1}{4}F^{\mu\nu}F_{\mu\nu} +\partial^\mu\phi^*\partial_\mu\phi +iqA^\mu\left(\phi^*\partial_\mu\phi-\phi\partial_\mu\phi^*\right)+q^2\left|\phi\right|^2A^\mu A_\mu+\mu_s \left|\phi\right|^2-\lambda_s\left|\phi\right|^4\nonumber\\
&\,+\frac{1}{2}\left(\frac{2q^2}{m_p^4}\left|\phi\right|^2+a_0^2\right)\partial_\mu F^{\mu\nu}\partial_\rho F^{\rho}_{\phantom{\rho}\nu},\label{LH}
\end{align}
where the first line is the one expected from the usual Maxwellian Higgs model.

In order to proceed, we restrict our analysis to the parameter region in which $\mu_s>0$ and we write the complex scalar field in the polar representation:

\begin{align}
\phi\left(x\right)=&\,e^{-i\frac{\theta\left(x\right)}{v}}\left[\frac{\chi\left(x\right)+v}{\sqrt{2}}\right],
\end{align}
where $\theta$ and $\chi$ are real scalar fields and $v$ is the constant

\begin{equation}
v=\sqrt{\frac{\mu_s}{\lambda_s}}.
\end{equation}

In terms of these newly defined scalar fields, the Lagrangian density (\ref{LH}) takes the form

\begin{align}
\label{genlaghiggs}
\mathcal{L}_H=&\,-\frac{1}{4}F^{\mu\nu}F_{\mu\nu}+\frac{1}{2}\left(a_0^2+\frac{q^2v^2}{m_P^4}\right)\partial_\mu F^{\mu\nu}\partial_\rho F^\rho_{\phantom{\rho}\nu}+\frac{q^2v^2}{2}\left(A^\mu+\frac{1}{qv}\partial^\mu\theta\right)\left(A_\mu+\frac{1}{qv}\partial_\mu\theta\right)\nonumber\\
&\,+\frac{1}{2}\partial^\mu\chi\partial_\mu\chi-\mu_s\chi^2-\mathcal{V}_I\left[A,\theta,\chi\right],
\end{align}

Here, $\mathcal{V}_I\left[A,\theta,\chi\right]$ is a term of interaction among the three fields, whose exact form is unimportant for our purposes.\footnote{As a matter of fact, $\mathcal{V}_I\left[A,\theta,\chi\right]$ is the same interaction potential for the Maxwellian Higgs model added to the following two terms: $\left(q^2/2m_P^4\right)\chi^2\partial_\mu F^{\mu\nu}\partial_\rho F^\rho_{\phantom{\rho}\nu}$ and $\left(q^2v/m_P^4\right)\partial_\mu F^{\mu\nu}\partial_\rho F^\rho_{\phantom{\rho}\nu}$. So, the $\theta$ field remains massless, in accordance to the Golstone theorem, while the Lagrangian density presents a term of the form $\left(M^2/2\right)A^\mu A_\mu$. Such term is associated, in first-order derivative theories, with the non-vanishing mass of the gauge field.} More importantly, we notice that if we define

\begin{align}
&q^2v^2\equiv\,M^2\equiv m_s^2\\
&a^2\equiv\,a_0^2+\frac{q^2v^2}{m_P^4}\equiv\frac{1}{m^{2}},
\end{align}
in which we have an infrared (IR) mass $M$ and an ultraviolet mass (UV) $m$ in term of the Podolsky parameter. Interesting the (IR) and the (UV) regime are tied by the interaction. In the free case ($q=0$) we have (5+2) degrees of freedom (Podolsky+complex scalar = Maxwell+Proca+complex scalar) but when we turn on the interaction we have (3+3+1) degrees of freedom (Proca+Proca+Higgs). Therefore, we see how the generalized photon eats the Goldstone boson, shifting the Podolsky parameter and gaining an extra longitudinal degree of freedom. The first three terms of the previous Lagrangian density are exactly equal to the Lagrangian density for the generalized Stueckelberg model (\ref{Lagra}). This allows us to apply the analysis presented in the previous sections and analogously quantize the generalized Higgs model by writing the transition amplitude in covariant form:

\begin{align}
Z=\int \left[\prod_{\mu=0}^3\mathcal{D}A^\mu\right]\mathcal{D}\theta\mathcal{D}\chi\mathcal{D}\bar{c}\mathcal{D}c\, e^{i\int d^4x\mathcal{L}_{{e}}},
\end{align}
where the effective Lagrangian density for the generalized Higgs model $\mathcal{L}_{{e}}$ is defined as

\begin{align}
\mathcal{L}_{{e}}\equiv \mathcal{L}_H-\frac{1}{2\xi}\left[\left(1+\frac{\Box}{m^2}\right)\partial_\mu A^\mu-gM\theta\right]^2 +\bar{c}\left[\left(1+\frac{\Box}{m^2}\right)\Box+gM ^2\right]c,
\end{align}
where $\bar{c}$ and $c$, like in the Stueckelberg model, are Grassmannian scalar fields and $g$ is a real parameter. The theorem (\ref{theorem}) allows us to choose $g$ as we please. So if we choose $g=\xi$ and take the limit of $m\rightarrow\infty$ we decouple the Goldstone-photon interaction. On the other hand if we want to decouple the Goldstone-photon interaction term independent of the previous $m$ limit, we need to generalize the covariant derivative of the Higgs Lagrangian term $\nabla_\mu\phi\left(\nabla^\mu\phi\right)^*$ in the following form

\begin{align}
\nabla_{\mu}\equiv&\, (1+\frac{\Box}{m_{}^{2}})\partial_\mu -iq\left(A_\mu+iZ_\mu\right),
\end{align}
in a way similar to the Stueckelberg analysis in eq. (\ref{decoupled}). Generalizing the Stueckelberg and Higgs fields with higher-derivatives has implications outside the investigation present here \cite{Donohue}, though, creating a \emph{per se} atmosphere for future investigations.


As we can see, the Fadeev-Popov trick in the Stueckelberg form helps to write the transition amplitude in the t'Hooft's shape \cite{thooft, Framp}. The demonstration that the transition amplitude and, reciprocally, the physics does not depend on $g$ is the same as Sec. \ref{BRS-T}. Other implication of BRS-T symmetry study is that we can generate the Ward-Takahashi identities, preparing the ground for a future demonstration that the theory is renormalizable independent of the spontaneous symmetry breaking \cite{R1,R2,R3,R4,R5,R6}.

\section{Conclusions and final remarks}
\label{conclusion}

This paper is devoted to analyse some mechanisms of mass generation in the context of generalized Stueckelberg-Higgs gauge theories. Our goal was to try to achieve a better understanding of the ultraviolet and the infra-red regime of the photon respecting gauge symmetry. The inspiration of this study lies in understanding the relationship between masses, size of the particles\textcolor{blue}{,} and range of the interactions.

Firstly we work with the generalized Stueckelberg electrodynamics studying the classical equation of motion, constructing the dynamics in the phase space, dealing with the constraints with the Dirac methodology, implementing the gauge fixing condition, constructing the transition amplitude with the Fadeev-Senjanovic procedure and writing it into a covariant form with the Fadeev-Popov method. It is interesting to see the link between the constraints and the physical degrees of freedom in the quantization process \cite{Anderson}. As we can see we have 18 variables and 6 constraints. So we have 12 variables describing the phase space and 6 degrees of freedom, Maxwell+Proca+Scalar (2+3+1). The present model can be seen as a toy model to study the quantization process in the functional integration approach.

Secondly when we write the transition amplitude in a covariant form  we can see the emergence of a  t'Hooft parameter $g$, eq. (\ref{tampli}). We prooved that this parameter does not contribute to the physics by using BRS-T symmetry of the transition amplitude of the theory, following the same line of reasoning of \cite{BRS-T, BRS-T2}. So the generalized Stueckelberg gauge theory is one more example that shows the link between the fictitious parameters and BRS-T symmetry.

Thirdly as a complement of the discussion, following the historical influence of the Stueckelberg field to the Higgs mechanism and generation of mass \cite{Stueck}, we construct the generalized Higgs model. As can been seen in eq. (\ref{genlaghiggs}) the Podolsky parameter depends of the interaction coupling $q$, where we define the infrared (IR) mass and the ultraviolet mass (UV). So, in this view, the Podolsky parameter's actual value is a consequence of the interaction with the Higgs field. The interesting thing here is that in the language of spontaneous symmetry breaking the generalized photon eat the Goldstone boson and acquire an IR and UV masses, therefore this regimes are tied by the interaction. As a consequence, we have the tools to study the influence of the radiative corrections in the generalized Stueckelberg-Higgs theory, exploring the proprieties of the effective action in this general context \cite{Coleman, Peskin, Weinb}.

As we have seen, the physical sector of the propagator of the Podolsky model with the Proca term is proportional to the difference of two massive propagators. It suggests that, differently from the Podolsky model, it must propagate $6$ and not just $5$ degrees of freedom. In order to properly answer this question, a rigorous Ostrogradsky analysis of the constraints of this model was considered. This was one of the motivation of our work. Regarding the inclusion of the Stueckelberg particle, it allows a gauge invariant description which implies in the possibility of useful Ward identities that may reveal non perturbative results associated to the model \cite{WTStu}. As we see, the Stueckeberg particle does not change the system's degree of freedom since it is pure gauge and lies outside of the physical spectrum. Its inclusion turn the constraint analysis more involved but preserves a first-class nature interesting enough to be explored here. Moreover, there is the possibility to associate this structure with a generalized Higgs mechanism, one of the main goals of this discussion.

Finally the next steps of this research would be introduce the coupling with fermions and to investigate how to implement the Podolsky photon in the context of electroweak theory $SU(2)\times U(1)$. We think that the relation between the neutral boson $Z_{0}$ and the photon $A_{\mu}$ are generalized in the sense of Podolsky with the Stueckelberg and Higgs masses, so we will have a better UV and IR controlling of divergences. Other context that the Stueckelberg field could be applied is in $SU(3)\times U(1)$ model for electroweak interactions \cite{Pleitez} where the appearance of more phenomenological vertices would be interesting. These matters will be further worked out and require elaborations.

\section*{Acknowledgement}
G. B. de Gracia thanks CAPES PhD grant (CP), A. A. N. thanks National Post-Doctoral Program grant (PNPD) for support and B. M. P. thanks CNPq for partial support.

\section*{Compliance with Ethical Standards}

{\bf Funding}: Capes, PNPD, Cnpq. {\bf Conflict of Interest}: There are no known conflicts of interest associated with this publication and there has been no significant financial support for this work that could have influenced its outcome. {\bf Ethical Conduct}: The article is aligned with ethical principles and responsible conduct of research.


\begin{thebibliography}{999}

\bibitem{Uti}R. Utiyama, Phys. Rev. 101, 1597 (1956).

\bibitem{Rai}L. O’Raifeartaigh, The Dawning of Gauge Theory, (New Jersey, Princeton University Press, 1997).

\bibitem{Leite}J. Leite Lopes, Gauge Field Theories: An Introduction,  (New York, Pergamon Press, 2010).

\bibitem{Yukawa} H. Yukawa, Proc. Phys. Math. Soc. Japan. 17, 48 (1935).

\bibitem{Proca}A. Proca, J. Phys. Radium 7,347 (1936); C. R. Acad. Sci. Paris 202, 1366 (1936).

\bibitem{Chinese}Liang-Cheng Tu, Jun Luo1 and George T Gillies, Rep. Prog. Phys. 68, 77 (2005).

\bibitem{Stuck}E. C. G. Stueckelberg , Phys. Rev. 52, 41 (1937); Helv. Phys. Acta 11, 225 (1938); Helv. Phys. Acta 11, 299 (1938); Helv. Phys. Acta 11, 312 (1938).

\bibitem{Stueck} H. Ruegg and M. Ruiz-Altaba,
International Journal of Modern Physics A 19, 3265 (2004).

\bibitem{Indians} T.R. Govindarajan and N. Kalyanapuram,   arXiv:1902.08768 (2019).

\bibitem{Higgs} P. W. Higgs, Phys. Rev 13, 16 (1964).

\bibitem{Coleman}S. Coleman and E. Weinberg, Phys. Rev. D. 7, 1888 (1973); R. Jackiw, Phys. Rev. D 9, 1686 (1974); S. Coleman, R. Jackiw and H. D. Politzer, Phys. Rev. D 10, 2491 (1974).

\bibitem{Peskin}M. E. Peskin and D. V. Schroeder, An Introduction to Quantum Field Theory, (Perseus Books, New York, 1995).

\bibitem{Weinb} S. Weinberg, The Quantum Theory of Fields vol. 2: Modern Applications, (Cambridge University Press, Cambridge, 1996).

\bibitem{Rubakov} V. Rubakov, Classical Theory of Gauge Fields,
(Princeton University, New Jersey Press, 2002).

\bibitem{Weinberg} S. Weinberg, A Model of Leptons, Phys. Rev 19, 21
(1967).

\bibitem{HooftV}G. 't Hooft and M. Veltman, Nucl. Phys. B 44, 189 (1972).

\bibitem{Nambu} Y. Nambu and G. Jona-Lasinio, Phys. Rev. 122, 345 (1961); Phys. Rev. 124, 246 (1961);  V. G. Vaks and A. I.  Larkin, Sov. Phys. JETP. 13, 192 (1961).

\bibitem{Nakanishi}N. Nakanishi and I. Ojima, Covariant Operator Formalism of Gauge Theories and Quantum Gravity, 1st edn. (World Scientific, Singapore, 1990).

\bibitem{Ebert}D. Ebert, H. Reinhardt, and M. Volkov, Prog. Part. Nucl.
Phys. 33, 1 (1994); M. Creutz, From quarks to pions: Chiral symmetry and confinement, (World Scientific, Singapore, 2018).

\bibitem{Bonin}B. M. Pimentel, L. A. Manzoni, J. L. Tomazelli, Eur. Phys. J. C 8, 353 (1999); B. M. Pimentel, J. T. Lunard and L. A. Manzoni, Int. J. Mod. Journal of Physics A, 15, 3263 (2000); B. M. Pimentel, L. A. Manzoni, J. L. Tomazelli, Eur. Phys. J. C 12, 701 (1999); R. Bufalo, R. Casana, B.M. Pimentel, Int. J. Mod. Phys. A 26, 1545 (2011); C. A. Bonin, B. M. Pimentel,  Eur. Phys. J. Plus 134, 533 (2019).

\bibitem{LeB}M. Le Bellac, Thermal Field Theory, 1st edn. (Cambridge University Press, Cambridge, 1996).

\bibitem{JD}S. Weinberg, Phys. Rev. D 9, 3357 (1974); L. Dolan and R. Jackiw, Phys. Rev. D 9, 3320 (1974).

\bibitem{Bop} F. Bopp, Ann. Phys. (Leipzig) 430, 345 (1940);
B. Podolsky, Phys. Rev. 62, 68 (1942); B. Podolsky and C.
Kikuchy, Phys. Rev. 65, 228 (1944); B. Podolsky and P. Schwed,
Rev. Mod. Phys. 20, 40 (1948).

\bibitem{Frank} J. Frenkel, Phys. Rev. E 54, 5859
(1996); A. E. Zayats, Ann. Phys. 342, 11 (2014).

\bibitem{LW2} F.A. Barone, G. Flores-Hidalgo and A.A. Nogueira, Phys. Rev. D 91, 027701 (2015); F. A. Barone, G. Flores-Hidalgo, and A. A. Nogueira, Phys. Rev. D 88, 105031 (2013).

\bibitem{PV} W. Pauli and F. Villars, Rev. Mod. Phys. 21, 434 (1949); J. Rayski, Acta. Phys. Pol. 9, 129(1948); Phys. Rev. 75, 1961 (1949).

\bibitem{RThibes}C. R. Ji, A. T. Suzuki, J. H. Sales and R. Thibes, Eur. Phys. J. C 79, (2019).

\bibitem{Daniel} R. Bufalo, B. M. Pimentel and D. E. Soto, Phys.
Rev. D 90, 085012 (2014).

\bibitem{Brandt} F. T. Brandt, J. Frenkel and D. G. C. McKeon, Mod.
Phy. Lett. A 31, 32 (2016).

\bibitem{Anderson}A. A. Nogueira, B. M. Pimentel and L. Rabanal, Nucl. Phys. B 934, 665 (2018); A. A. Nogueira, C. Palechor and A. F. Ferrari,  Nucl. Phys. B 939, 372 (2018).

\bibitem{LeeW} T. D. Lee and G. C. Wick, Nucl. Phys. B 9,
209 (1969); Phys. Rev. D 2, 1033 (1970); A. Accioly, P. Gaete, J. H. Neto, E. Scatena and R. Turcati, Mod. Phys. Lett. A
26, 26 (2011).

\bibitem{FABAAN1} F. A. Barone and A. A. Nogueira, Eur. Phys. J. C 75, 339 (2015).

\bibitem{Blaz}M. Blazhyevska, Journal of Physical Studies 16, 3001 (2012).

\bibitem{FabAnd}F. A. Barone and A. A. Nogueira, International Journal of Modern Physics: Conference Series 41, 1660134, (2016).

\bibitem{BufPimSot}R. Bufalo, B. M. Pimentel and D. E. Soto  Inter. Jour. Mod. Phys. A, 32, 1750165 (2017).

\bibitem{Russos} D. S. Kaparulin, S. L. Lyakhovich, and A. A. Sharapov, Eur. Phys. J. C 74, 3072 (2014); Jialiang Dai, Eur. Phys. J. Plus 135, 555 (2020).

\bibitem{Renorma}R. Bufalo, B. M. Pimentel, and G. E. R. Zambrano,
Phys. Rev. D 83, 045007 (2011). R. Bufalo, B. M. Pimentel, and G. E. R. Zambrano, Phys. Rev. D 86, 125023 (2012). R. Bufalo and B. M. Pimentel, Phys. Rev. D 88, 065013 (2013); R. Bufalo, T. R. Cardoso, A. A. Nogueira and B. M. Pimentel, Phys. Rev. D 97, 105029 (2018).

\bibitem{Galvao} C. A. P. Galv\~{a}o and B. M. Pimentel, Can. J. Phys. 66, 460
(1988).

\bibitem{Deb}C. A. Bonin, G. B. de Gracia, A. A. Nogueira and B. M. Pimentel, Int. J. Mod. Phys. A 35, 2050179 (2020).

\bibitem{BRS-T}A. A. Nogueira and B. M. Pimentel, Phys. Rev. D 95, 065034 (2017).

\bibitem{BRS-T2}G. B. de Gracia, B. M. Pimentel and L. Rabanal, Nucl. Phys. B, 114750 (2019).

\bibitem{Sto}L. H. C.Borges, F. A. Barone, C. A. M. de Melo, F. E. Barone, Nucl. Phys. B 944, 114634 (2019).

\bibitem{Noether}E. Noether, Nachr. d. Knig. Gesellsh. d. Wiss. Zu Gttingen 235 (1918). (Math. Phys. Klasse); English translation by M. A. Tavel, Transport Theory Statistical Phys. 1, 183 (1971).

\bibitem{Barut} A. O. Barut, Electrodynamics and Classical Theory of Fields and Particles, Dover Publications; Revised edition (2010).

\bibitem{Ostrogradsky} M. Ostrogradsky. Acad. de St. Pet. VI, 4, 385 (1850); R. Weiss. P. roc. R. Soc. London, A, 129, 102 (1938); J. S. De Wet. Proc. Cambridge Philos. Soc. 43, 51 1 (1947); 44, 526 (1948); T. S. Chang. P. roc. Cambridge Philos. Soc. 44, 76 (1948).

\bibitem{Dirac} Paul A. M. Dirac, Lectures on Quantum Mechanics, Snowball Publishing (2012).

\bibitem{Henneaux Teitelboim} Marc Henneaux \& Claudio Teitelboim, Quantization of Gauge Systems, Princeton University Press; Second Printing edition (1994).

\bibitem{Gribov}D. Zwanziger, Nuclear Physics B 323, 513 (1989).

\bibitem{Schwartz}M. Schwartz, Quantum Field Theory and The Standart Model, (Cambrige University Press, New York, (2014)).

\bibitem{Constraint}K. Sundermeyer, Constrained Dynamics, 1st ed. (Springer-Verlag, Berlin, 1982); H. J. Rothe and K. D. Rothe, Classical and Quantum Dynamics of Constrained Hamiltonian Systems, 1st ed. (World Scientific, New Jersey, 2010).

\bibitem{Zambrano}German E. R. Zambrano and  Bruto M. Pimentel,
Momento, Revista de Física 56, 26 (2018).

\bibitem{Nash}C. Nash, Relativistic Quantum Fields, 1st edn. (Dover, New York, 2011).

\bibitem{Smatrix}B. L. Voronov, P. M. Lavrov and I. V. Tyutin, Yad. Fiz.36, 498 (1982).

\bibitem{Framp}P. H. Frampton, Gauge Field Theories, 3rd edn. (Wiley-VCH, Weinheim, 2008).

\bibitem{thooft} G. ’t Hooft, Nucl. Phys. B 33, 173 (1971); Nucl. Phys. B 35, 167 (1971).

\bibitem{R1}G. ‘t Hooft and M. Veltman, Nucl. Phys. B 44, 189 (1972).

\bibitem{R2}K. Fujikawa. B. W. Lee and A. I. Sanda, Phys. Rev. D 6, 2923 (1972).

\bibitem{R3}B. W. Lee and J. Zinn-Justin Phys. Rev. D 5, 3121 (1972); Phys. Rev. D 5, 3137(1972);
 Phys. Rev. D 5, 3155(1972); Phys. Rev. D 7, 1049 (1973).

\bibitem{R4}J. J. Strathdee and A. Salam, Nuovo Cimento Ser.11, 11 A, 397 (1972).

\bibitem{R5}D. A. Ross and J. C. Taylor, Nucl. Phys. B 51, 125 (1973).

\bibitem{R6}K. Nishijima, M. Okawa, Progress of Theoretical Physics 61, 1822 (1979).

\bibitem{WTStu}Hendrik van Hees, The renormalizability for massive Abelian gauge field theories re-visited, arXiv:hep-th/0305076 (2003).

\bibitem{Donohue}John F. Donoghue and Gabriel Menezes
Phys. Rev. D 104, 045010 (2021).

\bibitem{Pleitez}F. Pisano and V. Pleitez, Phys. Rev. D. 46, 410 (1992).



\end{thebibliography}
\end{document}